\documentclass[10pt,a4paper]{article}
\usepackage[utf8]{inputenc}
\usepackage[english]{babel}
\usepackage{amsthm}
\usepackage{amsmath}
\usepackage{amsfonts}
\usepackage{amssymb}
\usepackage{braket}
\usepackage{graphicx}
\usepackage{mathrsfs}
\usepackage{verbatim}
\usepackage{enumitem}
\usepackage{bbm}

\theoremstyle{plain}
\newtheorem{thm}{Theorem}[section]
\newtheorem{cor}[thm]{Corollary}

\newtheorem{prop}[thm]{Proposition}

\theoremstyle{definition}
\theoremstyle{definition}
\newtheorem{defin}[thm]{Definition}
\theoremstyle{remark}

\newcommand{\id}{\mathbbm{1}}

\newcommand{\ldue}{\ell ^2}
\newcommand{\linf}{\ell ^\infty}
\newcommand{\luno}{\ell ^1}

\newcommand{\Lqt}[1]{\mathcal{L}^{#1}_{t}}
\newcommand{\Lrx}[1]{\mathcal{L}^{#1}(\Z^2)}

\newcommand{\Z}{\mathbb{Z}}

\newcommand{\R}{\mathbb{R}}
\newcommand{\C}{\mathbb{C}}
\newcommand{\T}{\mathbb{T}}

\newcommand{\norm}[1]{\left|\left| #1 \right|\right|}

\def\cH{\mathcal{H}}

\begin{document}
	\title{Absolute continuity of the spectrum of coupled
		identical systems on 1D  lattices}
	%
	
	\author{B. Langella$^{*}$,
		D. Bambusi\footnote{Dipartimento di Matematica, Universit\`a
			degli Studi di Milano, Via Saldini 50, I-20133
			Milano. \newline \textit{Email: }
			\texttt{beatrice.langella@unimi.it, dario.bambusi@unimi.it}},} \def\bv{{\bf v}}
	
	\maketitle
	
	\begin{abstract}
		We prove that the spectrum of the discrete
                Schr\"odinger operator on $\ell^2(\Z^2)$
\begin{align}
\label{hamiltoniana astratta1}
		(\psi_{n,m})\mapsto -(\psi_{n+1,m} +\psi_{n-1,m} +
		\psi_{n,m+1} +\psi_{n,m-1})+V_n\psi_{n,m} \	, \\
\nonumber
\quad (n, m)
		\in \Z^2,\ \left\{ V_n\right\}\in\ell^\infty(\Z)
\end{align}
is absolutely continuous.
	\end{abstract}
		
	\section{Introduction} \label{intro}
	
	In this paper we study the spectrum of the discrete Schr\"odinger operator on $\ldue(\Z^2)$
	\begin{equation} \label{hamiltoniana astratta}
	H:=	-\Delta + V,
	\end{equation}
	where $-\Delta$ is the discrete Laplacian acting on
        $\ell^2(\Z^2)\ni \psi = \{\psi_{n,m}\}$ by
	$$
	\left(-\Delta \psi \right)_{n,m} = -(\psi_{n+1,m} +\psi_{n-1,m} + \psi_{n,m+1} +\psi_{n,m-1}) \quad \forall\ (n, m) \in \Z^2
	$$
	and $V=\{V_n\}_{n \in \Z}$ is a sequence of real numbers defining a multiplication operator by
	$$
	\psi_{n,m} \mapsto V_n \psi_{n,m} \quad \forall\ (n,m) \in \Z^2.
	$$
	We emphasize that $V$ is independent of $m,$ so that the
        operator \eqref{hamiltoniana astratta} can be interpreted as
        describing infinitely many identical chains, each one labeled
        by the index $m$ and coupled to the others by a discrete
        laplacian along the $m$ direction.\\
	Our main result is the following one:
	\begin{thm} \label{intro main result}
	Let $V=\{V_n\}_{n \in \Z}\in \ell^\infty(\Z);$ then \eqref{hamiltoniana astratta} has absolutely continuous spectrum.
	\end{thm}
	
	We will also prove the following dispersive estimate:
	\begin{thm} \label{intro dispersive}
	Under the hypotheses of Theorem \ref{intro main result}, let $\psi_0 \in \luno(\Z^2);$ then $\exists\ C>0$ such that 
	\end{thm}
	\begin{equation} \label{intro dispersive est}
	\|e^{i\left(-\Delta + V\right)t} \psi_0\|_{\linf(\Z^2)} \leq \frac{C} {\langle t \rangle^{\frac{1}{3}}} \| \psi_0\|_{\luno(\Z^2)} \quad \forall\  t \in \R.
	\end{equation}
	Of course from Theorem \ref{intro dispersive} one can deduce standard Strichartz estimates (see \cite{keel_tao}).\\
	The main examples we have in mind are the case of a quasi-periodic potential, where ${V_n:= a V(\omega
		n+\theta)} $, with $V\in C^\infty(\T^1)$ and $a \in \R,$ and the case where $\{V_n\}_{n \in \Z}$ is a random sequence. For both cases it is known that,
	under suitable assumptions, the spectrum of the unidimensional operator
	\begin{equation} \label{qper 1d}
	\psi_n\mapsto -(\psi_{n+1}+\psi_{n-1})+ V_n \psi_n, \quad \psi=\{\psi_n\}_{n \in \Z} \in \ldue(\Z)
\end{equation}
is pure point. See \cite{fro_spe_way} for the latter and for instance
\cite{jitom} in the case of the almost- Mathieu operator, namely
$V(\theta) = a \cos(\theta),$ and \cite{eli_discr} for the proof of
pure pointness under more generic hypotheses on the quasi-periodic
potential $V.$\\ Theorem \ref{intro main result} shows that coupling
infinitely many chains of type \eqref{qper 1d}, the spectrum becomes
absolutely continuous, independently of the spectral nature of the
operator in \eqref{qper 1d}. Furthermore, by Theorem~\ref{intro
  dispersive}, one gets dispersion.  On the contrary, as shown in
\cite{eli_discr_bis} for quasi-periodic potentials, pure pointness
persists if the number of chains we couple is finite.\\ The proof of
our results is very easy; nevertheless, we think that they could help
to clarify the behavior of two dimensional chains. We recall that if,
instead of considering infinitely many {\it identical} coupled chains,
one considers a situation in which $V$ is a quasi-periodic potential
depending in a nondegenerate way on $m$ too, system
\eqref{hamiltoniana astratta} exhibits Anderson localization, as shown
in \cite{bou_gol_sch} in the two-dimensional case and in
\cite{bourg_hd} in higher dimensions, so that the spectrum is pure
point.  We also remark that our method trivially extends to higher
dimensional lattices.

The rest of the paper is devoted to the proof of the above results.

\noindent{\it Acknowledgements.} We thank Didier Robert for pointing to
our attention the paper \cite{spec_var_sep}.

\section{Absolute Continuity of the Spectrum} \label{proof_ass_cont}
We exploit the main result of \cite{spec_var_sep} dealing with
selfadjoint operators in Hilbert spaces $\cH$ with the structure of
tensor product $\cH=\cH_1\otimes \cH_2$. 

First we recall the definition of separable operators.

\begin{defin}
\label{sep}
A selfadjoint operator $H$ on $\cH$ is called separable with parts
$A_1$, $A_2$, if there exist cores $D_i\subseteq \cH_i$ such that the linear span of
$D_1\otimes D_2$ is also a core for $\cH$ and, for all pure tensors
$\psi=\psi_1\otimes \psi_2$ with $\psi_i\in D_i$, one has
$$
H \psi = A_1 \psi_1 \otimes \psi_2 + \psi_1 \otimes A_2 \psi_2 .
$$ 
\end{defin} 

By \cite{spec_var_sep} the following result holds.

\begin{thm} \label{teo mvp tensore}
	Let $H$ be a separable selfadjoint operator on $\mathcal{H}$,
        then the projection valued measure $P$ associated to $H$ is
        the {tensor convolution measure,} defined by the
        relation
	\begin{equation} \label{mvp tens}
	\langle P(\ \cdot\ ) \psi_1 \otimes \psi_2, \phi_1 \otimes \phi_2 \rangle_{\mathcal{H}} = \langle P_1(\ \cdot\ ) \psi_1, \phi_1 \rangle_{\mathcal{H}_1} \ast \langle P_2(\ \cdot\ ) \psi_2, \phi_2 \rangle_{\mathcal{H}_2},
	\end{equation}
	where $P_i$ are the projection valued measures associated to $A_i,\ i=1,\ 2.$
\end{thm}

We will apply this theorem to our case exploiting the structure
$\ldue(\Z^2) = \ldue(\Z) \otimes \ldue(\Z)$ of the space. Indeed one
immediately sees that the operator \eqref{hamiltoniana astratta} is
separable with parts
$$
(A_1 \chi)_{n} = -(\chi_{n+1}+\chi_{n-1})+ V_n \chi_{n}, \quad \chi=\{\chi_{n}\}_{n\in \Z} \in \ldue(\Z)
$$ and $A_2\equiv\Delta_m$ the discrete laplacian in the $m$
direction, namely
$$
\left(-\Delta_m \phi\right)_{m}= -\left(\phi_{m+1} + \phi_{m-1}\right) \quad \phi=\{\phi_{m}\}_{m\in \Z} \in \ldue(\Z).
$$

In order to exploit Theorem \ref{teo mvp tensore}, we first give a
result on the convolution of two bounded measures, one of which is
absolutely continuous.

\begin{prop} \label{ass cont conv}
	Let $m,\ n$ be complex finite Borel measures.  If $m$ is
        absolutely continuous with respect to Lebesgue measure, then
        their convolution $m \ast n$ is absolutely continuous with
        respect to Lebesgue measure.
\end{prop}
	 
\begin{proof}
Let $f \in \mathcal{L}^1(d (m n)),$ with $m n$ the product measure.
One has
\begin{align*}
	 \int_{\R} f d(m\ast n) &= \int_{\R^2} f(x+y)\ d m(x)\ d n(y)\\
	 &= \int_{\R} \int_{\R} f(z)\ d m(z-y)\ d n(y)\\
	 &= \int_{\R} \int_{\R} f(z)\ g(z-y)\ d z\ d n(y)\\
	 &= \int_{\R} f(z) \left(\int_{\R} g(z-y) d n(y)\right)\ d z,
\end{align*}
where ${\cal G}(z) = \int_{\R} g(z-y) d n(y) \in \mathcal{L}^1(dz),$ since
\begin{align*}
	\int_{\R} \left|\int_{\R} g(z-y) d n(y)\right|\ d z & \leq \int_{\R} \left(\int_{\R} |g(z-y)| \ d |n|(y)\right) \ d z\\
	&= \int_{\R} \left(\int_{\R} |g(z-y)| \ d z\right)\ d |n|(y)\\
	&=\int_{\R} \left(\int_{\R} |g(z)| \ d z\right)\ d |n|(y)\\
	&=\left(\int_{\R} |g(z)| \ d z\right) |n(\R)| < +\infty.
\end{align*}
\end{proof}

\begin{proof}[Proof of Theorem \ref{intro main result}]
	By density argument, it is sufficient to prove absolute continuity of the real Borel measure
	$\langle P (\ \cdot\ ) \psi, \psi \rangle_{\mathcal{H}}$
	for the set of all finite linear combinations of pure tensors,
	$$
	\psi = \sum_{j=1}^{N} \alpha_j \ \chi_j \otimes \phi_j,\quad N>0,\ \alpha \in \C,\ \chi_j,\ \phi_j \in \ldue(\Z).
	$$
	We apply Theorem \ref{teo mvp tensore} to get
	\begin{align*}
		\langle P (\ \cdot\ ) \psi, \psi \rangle_{\ldue(\Z^2)}
		&=\sum_{j, k=1}^{N} \overline{\alpha}_j \alpha_k \langle P (\ \cdot\ ) \chi_j \otimes \phi_j, \chi_k \otimes \phi_k \rangle_{\ldue(\Z^2)}\\
		&=\sum_{j, k=1}^{N} \overline{\alpha}_j \alpha_k 
		\langle P_1(\ \cdot\ ) \chi_j, \chi_k \rangle_{\ldue(\Z)} \ast \langle P_2(\ \cdot\ ) \chi_j, \chi_k \rangle_{\ldue(\Z)},
	\end{align*}
	where $P_1$ is the projection valued measure associated to $A_1$ and $P_2$ is associated to $-\Delta_m.$\\ 
	Thus we deal with a linear combination of the complex measures $\langle P_1(\ \cdot\ ) \chi_j, \chi_k \rangle_{\ldue(\Z)} \ast \langle P_2(\ \cdot\ ) \chi_j, \chi_k \rangle_{\ldue(\Z)},$ whose absolute continuity follows from Proposition \ref{ass cont conv}, and from the absolute continuity of the measures $\langle P_2(\ \cdot\ ) \chi_j, \chi_k \rangle_{\ldue(\Z)}$.
\end{proof}

\section{Dispersive Estimates}
In the following, if $p \in [1,\ \infty]$ we will denote with
$\ell^p_n$ (respectively, $\ell^p_m$) the $\ell^p$ space of a complex
valued sequence with respect to its integer index $n$ (respectively,
its $\ell^p$ space with respect to its integer index
$m$). Furthermore, $\ell^{p}_{n,m}$ will denote the norm of complex
sequence with respect to both the indexes $n,\ m.$

\begin{proof}[Proof of Theorem \ref{intro dispersive}]
Denote, by abuse of notation, $A_1= A_1\otimes \id$ and
$-\Delta_m=\id\otimes (-\Delta_m)$, then such operators strongly
commute and therefore one has
	$$
	e^{i (-\Delta+ V)t}= e^{i (-\Delta_m+ A_1)t} = e^{i A_1 t} e^{-\Delta_m t} \quad \forall \ t \in \R.
	$$
	The dispersive estimate for $-\Delta+V$ then follows from the analogous dispersive estimate for the one-dimensional discrete laplacian $-\Delta_m$ (see \cite{ste_kev}): let $\phi \in \ell^1(\Z);$ then $	\exists\ C>0$ such that
	$$
  \|e^{-i \Delta_m t}\phi\|_{\ell^{\infty}(\Z)} \leq C \langle t \rangle ^{-\frac{1}{3}} \|\phi\|_{\ell^{1}(\Z)}.
	$$
	Indeed, $\forall t \in \R$, if $\psi_0 \in \ell^1_{n,m}$ one has
	\begin{align*}
	\| e^{i (-\Delta+V)t} \psi_0 \|_{\ell^{\infty}_{n,m} }
	&=\|e^{i A_1t}\left( e^{-i \Delta_m t} \psi_0 \right) \|_{\ell^{\infty}_{n,m} }\\
	&\leq \norm{\norm{e^{i A_1t}\left( e^{-i \Delta_m t} \psi_0 \right)}_{\ldue_n}}_{\ell^{\infty}_m}\\
	&=\norm{\norm{ e^{-i \Delta_m t} \psi_0}_{\ldue_n}}_{\ell^{\infty}_m}\\
	&\leq \norm{ \norm{ e^{-i \Delta_m t} \psi_0 }_{\ell^{\infty}_m}}_{\ldue_n}\\
	& \leq C \langle t \rangle ^{-\frac{1}{3}} \norm{\norm{\psi_0}_{\ell^{1}_m}}_{\ldue_n}\\
	& \leq  C \langle t \rangle ^{-\frac{1}{3}} \norm{\psi_0}_{\ell^{1}_{n,m}}.
	\end{align*}
\end{proof}

\end{document}